\documentclass[twoside]{article}

\usepackage[accepted]{aistats2019}
\usepackage[round]{natbib}

\usepackage{subcaption}
\usepackage{hyperref}
\usepackage{amsmath,amsfonts,amssymb,amsthm, array, epsfig, epstopdf, url}
\usepackage{mathtools}
\usepackage{color}
\usepackage{dsfont}
\usepackage{algorithm}
\usepackage[noend]{algpseudocode}
\def\BState{\State\hskip-\ALG@thistlm}
\usepackage{algorithmicx}
\usepackage[bottom]{footmisc}

\newcommand{\pr}{\mathbb{P}}

\newcommand{\R}{\mathbb{R}}

\newcommand{\argmax}[1]{\underset{#1}{\arg\max}}
\newcommand{\argmin}[1]{\underset{#1}{\arg\min}}
\newcommand{\E}{\mathbb{E}}
\newcommand{\I}{\mathcal{I}}
\newcommand{\ind}{\mathds{1}}
\newcommand{\mle}{\hat{\theta}^{G}_{\text{MLE}}}

\usepackage[capitalize, sort]{cleveref}
\usepackage{thmtools}
\theoremstyle{plain}
\newtheorem{nthm}{Theorem}[section]
\newtheorem{nprop}[nthm]{Proposition}
\newtheorem{nlem}[nthm]{Lemma}

\theoremstyle{definition}
\newtheorem{ndefn}[nthm]{Definition}

\theoremstyle{remark}
\newtheorem{nrmk}{Remark}
\newtheorem*{rmk}{Remark}

\crefname{nlem}{Lemma}{Lemmas}
\crefname{nprop}{Proposition}{Propositions}
\crefname{ncor}{Corollary}{Corollaries}
\crefname{nthm}{Theorem}{Theorems}
\crefname{nexa}{Example}{Examples}
\crefname{ndefn}{Definition}{Definitions}
\crefname{nassum}{Assumption}{Assumptions}
\crefformat{footnote}{#1\footnotemark[#2]#3}

\begin{document}

%
\runningtitle{Targeted Causal Structure Discovery}

%
\runningauthor{Agrawal, Squires, Yang, Shanmugam, Uhler}

\twocolumn[

\aistatstitle{ABCD-Strategy: Budgeted Experimental Design \\for Targeted Causal Structure Discovery}

\aistatsauthor{ Raj Agrawal  \And Chandler Squires  \And  Karren Yang  \And Karthikeyan Shanmugam  \And \quad  Caroline Uhler }


\aistatsaddress{ MIT \And  MIT \And MIT \And MIT-IBM Watson AI Lab \\ IBM Research NY \And MIT}


]

\begin{abstract}
Determining the causal structure of a set of variables is critical for both scientific inquiry and decision-making. However, this is often challenging in practice due to limited \emph{interventional} data. Given that randomized experiments are usually expensive to perform, we propose a general framework and theory based on optimal Bayesian experimental design to select experiments for \emph{targeted causal discovery}. That is, we assume the experimenter is interested in learning some function of the unknown graph (e.g., all descendants of a target node) subject to design constraints such as limits on the number of samples and rounds of experimentation. While it is in general computationally intractable to select an optimal experimental design strategy, we provide a tractable implementation with provable guarantees on both approximation and optimization quality based on submodularity. We evaluate the efficacy of our proposed method on both synthetic and real datasets, thereby demonstrating that our method realizes considerable performance gains over baseline strategies such as random sampling. 
\end{abstract}

\section{Introduction} \label{sec:intro}
Determining the causal structure of a set of variables is a fundamental task in causal inference, with  widespread applications not only in artificial intelligence but also in scientific domains such as  biology and economics \citep{friedman00,pearl03,robins00,causality_book}. 
One of the most common ways of representing causal structure is through a \emph{directed acyclic graph} (DAG), where a directed edge between two variables in the DAG represents a direct causal effect and a directed path indicates an indirect causal effect \citep{causality_book}. 



Causal structure learning is intrinsically hard, since a DAG is generally only identifiable up to its \emph{Markov equivalence class} (MEC) \citep{verma90,andersson97}. Identifiability can be improved by performing \emph{interventions} \citep{IMEC2012,yang2018characterizing}, and several algorithms have been proposed for structure learning from a combination of observational and interventional data \citep{IGSP,IMEC2012,yang2018characterizing}. Since experiments tend to be costly in practice, a natural question is how principled \emph{experimental design} (i.e., selection of intervention targets) can be leveraged to maximize the performance of these algorithms under budget constraints.

Seminal works by \citet{koller_active} and \citet{murphy_active_learning} showed that experimental design can improve structure recovery in causal DAG models. 
However, these methods assume a basic framework in which experiments are performed one sample at a time. 
In practice, experimenters often perform a batch of interventions and collect samples over multiple rounds of experiments; and they must also factor in budget and feasibility constraints, such as on the number of unique interventions that can be performed in a single experiment, the number of experimental rounds, and the total number of samples to be collected. 
In genomics, for instance, genome editing technologies have enabled the collection of batches of large-scale interventional gene expression data~\citep{perturb_seq}. 
An imminent problem is understanding how to optimally select a batch of interventions and allocate samples across these interventions, over multiple experimental rounds in a computationally tractable manner. 

Since the initial works by \citet{koller_active} and \citet{murphy_active_learning}, there have been a number of new experimental design methods under budget constraints \citep{hauser_MEC,bed_MEC,bayes_sachs}. These methods suffer from two drawbacks: (1) poor computational scaling \citep[{cf.}][]{bayes_sachs} or (2) strong assumptions including the availability of infinite observational and/or interventional data from each experiment \citep[{cf.}][]{hauser_MEC,bed_MEC}.   
Since it is difficult to learn the correct MEC in a limited sample setting, it is desirable to use interventional samples not only to improve identifiability but also to help distinguish between observational MECs. 

Generalizing the frameworks in~\citep{koller_active,murphy_active_learning,berger,hauser_MEC,bayes_sachs}, we assume the experimenter is interested in learning some function $f(G)$ of the unknown graph $G$. Returning to gene regulation, one might set $f(G)$ to indicate whether some gene $X$ is downstream of some gene $Y$, i.e. if $X$ is a \emph{descendant} of $Y$ in $G$. Using targeted experimental design, all statistical power is placed in learning the target function rather than being agnostic to recovering all features in the graph. In addition, we also explicitly take into account that only finitely many samples are allowed in each round, and work under various budget constraints such as a limit on the number of rounds of experimentation.

We start by reviewing causal DAGs in \cref{sec:prelim} and then propose an entropy-based score function that generalizes the one by \citet{koller_active} and \citet{murphy_active_learning} in Section~\ref{sec_3}. Since optimizing this score function is in general computationally intractable, we propose our \emph{ABCD-Strategy} consisting of approximations via \emph{weighted importance sampling} and greedy optimization in Section~\ref{sec:tract_algo}. We also provide guarantees for this algorithm based on \emph{submodularity}. 
Further, in contrast to earlier score functions, we show that our proposed score function is provably \emph{consistent}. Finally, in \cref{sec:experiments} we demonstrate the empirical gains of the proposed method over random sampling on both synthetic and real datasets.  

\section{Preliminaries}
\label{sec:prelim}

\textbf{Causal DAGs:} \quad Let $G = ([p], A)$ be a \emph{directed acyclic graph} (DAG) with vertices $[p] := \{1,\ldots,p\}$ and directed edges $A$, where $(i,j)\in A$ represents the arrow $i\to j$. A \emph{linear causal model} is specified by a DAG $G$ and a corresponding set of edge weights $\theta\in\mathbb{R}^{|A|}$. Each node $i$ in $G$ is associated with a random variable $X_i$. Under the \textit{Markov Assumption}, each variable $X_i$ is conditionally independent of its nondescendants given its parents, which implies that the joint distribution factors as
$\prod_{i=1}^p \mathbb{P}\big(X_i \ | \ \mathsf{Pa}_G(X_i)\big),$ where $\mathsf{Pa}_G(X_i)$ denotes the parents of node $X_i$ \citep[Chapter 4]{causality_book}. This factorization implies a set of \textit{conditional independence} (CI) relations; 
 the \textit{Markov equivalence class} (MEC) of a DAG $G$ consists of all DAGs that share the same CI relations \citep[Chapter 3]{lauritzen_book}. The 
 \emph{essential graph} $\text{Ess($G$)}$ is a partially oriented graph that uniquely represents the MEC of a DAG by placing directed arrows on edges consistent across the equivalence class and leaves the other edges undirected \citep{andersson97}. 

\textbf{Learning with Interventions:} \quad Let  \emph{intervention} $I \subseteq [p]$ be a set of intervention targets. Intervening on $I$  removes the incoming edges to the random variables $X_{I} \coloneqq (X_i)_{i \in I}$ in $G$ and sets the joint distribution of $X_I$ to a new interventional distribution $\pr^I$. The resulting  \emph{mutilated graph} is denoted by $G^I$. 
A typical choice of $\pr^I$ is the product distribution $\prod_{i \in I} f_i (X_i)$, where each $f_i(X_i)$ is the probability density function for the intervention at $X_i$. We denote by $\I^* \coloneqq \{I_1, \cdots, I_K \}$ the set of all $K \in \mathbb{N}$ allowed interventions and by $\I \subseteq \I^*$ the subset of selected interventions. An intervention $I = \emptyset$ indicates observational data. We assume that $\I^*$ is a \emph{conservative family} of interventions, i.e., for any $i \in [p]$, there exists some $I_j \in \I^*$ such that $i \notin I_j$ \citep{IMEC2012}. Given a conservative family of targets $\I$, two DAGs $G_1$ and $G_2$ are \emph{$\I$-Markov equivalent} if they are observationally Markov equivalent and for all $I \in \I$, $G_1^I$ and $G_2^I$ have the same skeleta \citep{IMEC2012,IMEC2015}. The set of \emph{$\I$-Markov equivalent} DAGs can be represented by the \emph{$\I$-essential graph} $\text{Ess}^{\I}(G)$, a partially directed graph with at least as many directed arrows as $\text{Ess($G$)}$ \citep[Theorem 10]{IMEC2012}.

\textbf{Bayesian Inference over DAGs:} \quad In various applications, the goal is to recover a function $f(G)$ of the underlying causal DAG $G$ given a mix of $n$ independent observational and interventional samples $D = \{(X_{mi}, I^{(m)}): I^{(m)} \in \I^*, m\in [n], i\in [p]\}$. For example, we might ask whether an undirected edge $(i,j)$ is in $A$, or we might wish to discover which nodes are the parents of a node $i$. We can encode our prior structural knowledge about the underlying DAG through a \emph{prior} $\mathbb{P}(G)$.
The \emph{likelihood} $\pr(D \mid G)$ is obtained by marginalizing out $\theta$:
\begin{align*}
	\pr(D \mid G) &= \int_{\theta} \pr( D, \theta \mid G) \; d\theta \\
    	&= \int_{\theta} \pr(D \mid \theta, G) \pr(\theta \mid G) \; d\theta
\end{align*}
and can be computed in closed-form for certain distributions \citep{bge_orig,bge_score}. Applying Bayes' Theorem yields the \emph{posterior distribution} $\pr(G \mid D) \propto \pr(D \mid G) \pr(G)$, which describes the state of knowledge about $G$ after observing the data $D$. Given the posterior, we can then compute $\E_{\pr(G \mid D)} f(G)$, the posterior mean of some target function $f(G)$. Note that when $f$ is an indicator function, this quantity is a posterior probability.

\section{Optimal Bayesian Experimental Design}
\label{sec_3}
Our goal is to learn some feature $f(G)$ of the unknown graph through experimental design under budget constraints such as limited number of experimental rounds. In principle, this question can be answered using \emph{optimal Bayesian experimental design}, namely 
by selecting the experiment that maximizes the expected value of some \emph{utility function} $U$, where the expectation is with respect to hypothetical data generated according to our current beliefs \citep{bayes_exp_design}. Here, the expected utility function $U$ is a function defined on multisets of $\I^*$:     
\begin{ndefn} \label{def:exp_utility}
The \emph{expected utility} $U^f(\xi; D)$ of a multiset of interventions $\xi \in \mathbb{Z}^{\I^*}$ for learning a function $f(G)$ given currently collected data $D$ is given by
\begin{equation}
\begin{split}
U^f(\xi; D) &= \E_{y \sim \pr(y \mid D, \xi)} \ U^f(y, \xi; D) \\
		&= \E_{G, \theta \mid D} \E_{y \mid G, \theta, \xi} \ U^f(y, \xi; D), \quad y \in \R^{|\xi|}, 
\end{split}
\end{equation}
where $U^f(y, \xi; D) \in \R$ is a function measuring the utility of observing additional samples $y$ from a proposed design $\xi$ and $|\xi| \coloneqq \sum_{I \in \I^*} |\text{\# times $I$ in $\xi$}|$. The \emph{optimal Bayesian design} $\xi^*$ under a set of design constraints $C$ is given by
\begin{equation}
\label{opt_problem}
\xi^* \in \argmax{\xi \in  \mathbb{Z}^{\I^*} \cap  C} \ U^f(\xi; D).
\end{equation}
We denote samples collected from such an optimal strategy $\xi^*$ by $D_{\xi^*}$ 
\end{ndefn}
In \cref{def:exp_utility}, $y$ is distributed according to our current beliefs $\pr(y \mid D, \xi)$ = $\E_{G, \theta \mid D} \left[ \pr(y \mid G, \theta, \xi) \right]$, a \emph{mixture distribution} over $(G, \theta)$, and the utility function $U^f(y, \xi; D)$ is averaged over this distribution. There are many potential choices for $U^f(y, \xi; D)$, a popular one being \emph{mutual information}. \citet{koller_active}, \citet{berger} and \citet{murphy_active_learning} propose optimizing mutual information for the problem of recovering the full graph. More precisely, they consider the problem where $f(G) = G$ in the active learning setting, where the experimenter can adaptively collect one sample at a time. 
We here extend their framework to general functions $f(G)$ and the 
\emph{batched setting}, where multiple samples are collected at once and the total number of batches is fixed by the experimenter. Hence, $U^f$ must be defined on multisets instead of elements of $\I^*$ since multiple samples (i.e., interventions of the same type) may be collected in each batch. Note that the difficulty in solving~\cref{opt_problem} stems from the constraint set $C$, which renders this optimization problem combinatorial.  


Recently, \citet{bayes_sachs} proposed a Bayesian experimental design method to work in the batched setting. The authors proposed a utility function based on the expected number of additional edges that could be oriented by performing a particular intervention given the observational MEC. This function is similar to the one proposed by \citet{hauser_MEC} and \citet{bed_MEC}, in which interventions are chosen that fully identify the causal network given the MEC. Unfortunately, the algorithm in \citet{bayes_sachs} has \emph{factorial} dependence on the size of the batch; in addition, we prove in \cref{A:consis_counter} that their proposed utility function is, in general, not \emph{consistent}; see \cref{def:batch_consis} for a definition of  consistency.  

We therefore follow the approach taken by \citet{koller_active} and \citet{murphy_active_learning} and consider the utility function $U^f(y, \xi; D)$ to be given by mutual information. Maximizing the mutual information is equivalent to picking the set of interventions that leads to the greatest expected decrease in entropy of $f(G)$. The mutual information utility function is given by 
\begin{equation} \label{eq:mutual_info}
U_{\text{M.I.}}^f(y, \xi; D) \coloneqq H(f \mid D) - H(f \mid D, y=y, \xi),
\end{equation}
where the \emph{entropy} $H(f \mid D)$ equals
\begin{equation*}
\begin{split}
\sum_{e: f(G) = e}-\pr(f(G) = e \mid D) \log \pr(f(G) = e \mid D), \\ \ \text{and} \ \
\pr(f(G) = e \mid D) = \E_{\pr(G \mid D)} \ind(f(G) = e), \\
\pr(G \mid D) \propto \int_{\theta} \pr(D \mid G, \theta) \pr(\theta \mid G) \pr(G).
\end{split}
\end{equation*}
To better understand the behavior of $U_{\text{M.I.}}^f$, we prove the following proposition, which  highlights the behavior of $U_{\text{M.I.}}^f$ in the limit of infinite samples per intervention; this is the setting studied by \citet{hauser_MEC} and \citet{bed_MEC}. 
\begin{nprop} \label{prop:infinite_samps}
Suppose that the Markov equivalence class $\mathcal{G}$ of $G^*$ is known and the goal is to identify the underlying true DAG $G^*$. 
Furthermore, assume a uniform prior over $\mathcal{G}$, infinite samples per intervention $I \in \I$, and at most $K$ unique interventions per batch as in \citet{bed_MEC}. Then, $U_{\text{M.I.}}$ selects the interventions 
\begin{equation*}
\I_{\text{M.I.}} \in \argmin{|\I| \leq K} \ \frac{1}{|\mathcal{G}|} \sum_{G \in \mathcal{G}}  \log_2 |\text{Ess}^{\I}(G)|
\end{equation*}
%
%
%
where $|\text{Ess}^{\I}(G)| \coloneqq |\{G^{'} \in \mathcal{G}: G^{'} \in \text{Ess}^{\I}(G) \}|$. 
\end{nprop}
This result (proof in Appendix) shows that in the limiting case, mutual information selects interventions that lead to the finest expected log $\I$-MEC sizes. This limiting behavior of mutual information parallels what graph-based score functions do, such as the ones considered by \citet{hauser_MEC}, \citet{bed_MEC} and \citet{bayes_sachs}, that invoke the \emph{Meek Rules} \citep{meek_rules} to select interventions that orient the most number of edges in the $\I$-essential graphs (in expectation).   


A score function based on mutual information is particularly appealing since it not only has desirable properties in the infinite sample setting, but also does not require the MEC to be known, naturally handling the case of finite sample sizes. In particular, a score function based solely on Meek rules will not pick the same intervention twice by definition, since repeating the same intervention does not improve identifiability. As a result, adapting graph-based score functions in the finite sample regime requires first constructing an intervention set and then allocating samples instead of jointly picking and allocating samples. Mutual information, on the other hand, can pick the same intervention twice; for example if a particular intervention is very informative, selecting it twice and allocating more samples to it might lead to a greater expected decrease in entropy than a new intervention. 


\subsection{Budget Constraints}
\vspace{-0.1cm}
So far we have not specified the constraint set $C$ in \cref{opt_problem}. To this end, we assume that the experimenter has a total of $N$ samples to allocate across $B$ batches. While one could try to optimize the partition of $N$ samples across batches, in this work we study the simpler case where each batch $b$, $1 \leq b \leq B$, receives a pre-specified amount of samples $N_b$ with $\sum_{b}N_b=N$. For simplifying notation we assume throughout that $N_b = \frac{N}{B}$. 
We leave the study of adaptive batch sizes $N_b$ for future work. The constraint set then equals,
\begin{equation} \label{eq:our_constraints}
\begin{split}
C_{N, b} &\coloneqq \{\xi \in \mathbb{Z}^{\I^*}: |\xi| = N_b \},
\end{split}
\end{equation}
where the subscripts on $C$ emphasize the dependence on $N$ and $b$.
Then, the optimal design in batch $b$ is obtained by solving the following combinatorial optimization problem:
\begin{equation} \label{eq:intract_opt}
\xi^*_b \in \argmax{\xi \in  \mathbb{Z}^{\I^*} \cap C_{N, b}} \ U(\xi; D_{b-1}),
\end{equation}
where $D_{b-1} \coloneqq [D_{\xi^*_1}, \cdots, D_{\xi^*_{b-1}}]$ is all the data collected at the start of batch $b$ and $U(\xi; D_{b-1})$ could, for example, be the mutual information defined in \cref{eq:mutual_info}. Notice that while a particular form of $U(\xi; D_{b-1})$ is provided in \cref{def:exp_utility}, $U(\xi; D_{b-1})$ need not necessarily be a Bayesian utility function to fit within the framework of \cref{eq:intract_opt}.  

We now define a natural notion of consistency for any experimental design method that can be cast as an optimization routine in the form of \cref{eq:intract_opt}. Since the consistency of a utility function should not depend on a specific constraint set such as $C_{N, b}$, \cref{def:batch_consis} assumes the constraint set is arbitrary and set by the practitioner.  
\begin{ndefn} \label{def:batch_consis} 
Suppose $f(G)$ is identifiable in $\text{Ess}^{\I^*}(G^*)$, where $G^*$ is the true unknown DAG. Let $C_{N, b}$, $1 \leq b \leq B$ denote the constraints in batch $b$. 
A utility function $U(\xi)$ is \emph{budgeted batch consistent} for learning a target feature $f(G)$ if
\begin{equation*}
\pr \left( f(G) \mid D_{B} \right) \xrightarrow{\text{$\mu^*$ a.s.}} \ind(f(G) = f(G^*)), 
\end{equation*}
as $N, B \rightarrow \infty$, where $\mu^*$ is the law determined by the true unknown causal DAG $(G^*, \theta^*)$ 
\end{ndefn}

\begin{nthm} \label{thm:mutal_consistent}
$U_{\text{M.I.}}^f$ is budgeted batch consistent for \emph{single-node interventions}, i.e., when $\I^* = \{ \{1\}, \cdots, \{ p\} \}$.
\end{nthm}

\begin{nrmk}
While Theorem~\ref{thm:mutal_consistent} may not be surprising (proof in the Appendix), we found that various utility functions that seem natural and have been proposed in earlier work are not consistent in the budgeted setting. 
In particular, in \cref{A:consis_counter} we show that the utility function proposed by \citet{bayes_sachs} is not consistent for single-node interventions. The main issue is that there are DAGs and constraint sets for which the same interventions keep getting selected, instead of selecting new interventions to fully identify $f(G)$.    
\end{nrmk}

\section{Tractable Algorithm} \label{sec:tract_algo}
\vspace{-0.1cm}
While \cref{sec_3} provides a general framework for targeted experimental design, there are several computational challenges that we have not yet addressed. The first challenge is computing $U_{\text{M.I.}}^f(\xi; D)$. This objective function requires summing over an exponential number of DAGs and marginalizing out the edge weights $\theta$. In this section, we discuss how to approximate $U_{\text{M.I.}}^f(\xi; D)$ by sampling graphs (either through MCMC or the \emph{DAG-bootstrap} \cite{bootstrap_dag}) and using the maximum likelihood estimator of $\theta$ for each graph. Taken together, these approximations not only allow the mutual information score to be computed tractably but also lead to desirable optimization properties. In particular, we prove in \cref{thm:greedy_gaurantee} that our approximate utility function is submodular. This property enables optimizing the approximate objective in a  sequential greedy fashion with provable guarantees on optimization quality.  

\subsection{Expectation over $(G,\theta)$}
\vspace{-0.1cm}
A serious problem from a computational perspective is the expectation over $(G, \theta)$ in \cref{def:exp_utility}. Since the number of DAGs grows \emph{superexponentially} with $p$, enumerating all possible DAGs is intractable. Instead, in each batch $b$, we propose to sample $T$ graphs according to the posterior $\pr(G \mid D_{b-1})$. This can be done using a variety of different \emph{Markov chain Monte-Carlo} (MCMC) samplers; see for example \citet{high_prob_dags, ellis_wong,Friedman2003,partial_order,partition_mcmc,struct_mcmc,edge_reversal, minimap_paper}. An alternative that is often faster but still achieves good performance, is approximating the posterior via a high-probability candidate set of $T$ DAGs $\mathcal{\hat{G}}_T$~\citep{high_prob_dags,bootstrap_dag}. While there are many ways to build up this set, a popular approach is through the \emph{nonparametric DAG bootstrap}~\citep{bootstrap_dag}. The main idea is to subsample the data (with replacement) $T$ times and fit a DAG learning algorithm to each of the generated datasets to construct $\mathcal{\hat{G}}_T$. Each $G \in \mathcal{\hat{G}}_T$ can then be weighted according to the ratio of unnormalized posterior probabilities,
\begin{equation} \label{eq:post_weights}
w_{G, D} \coloneqq \frac{\pr(G) \pr(D \mid G)}{\sum_{G \in \mathcal{\hat{G}}_T} \pr(G) \pr(D \mid G)}
\end{equation}
to form an approximate posterior $\hat{\pr}(G) \coloneqq w_{G, D} \ind(G \in \mathcal{\hat{G}}_T)$. The DAG learning algorithm used for this purpose must be able to handle  a mix of observational and interventional data. Two recent methods that have been developed for this purpose are given in  \citet{IMEC2012} and \citet{IGSP}. We summarize constructing an approximate posterior via the DAG bootstrap in \cref{alg:sample_dags}.
\begin{algorithm} 
\caption{\texttt{DAGBootSample}}\label{alg:sample_dags}
\hspace*{\algorithmicindent} \textbf{Input:} N datapoints $D_N$, number of samples T  \\
\hspace*{\algorithmicindent} \textbf{Output:} $T$ bootstrap DAG samples $\mathcal{\hat{G}}_T$
\begin{algorithmic}[1]
\State $\mathcal{\hat{G}}_T \gets \emptyset$
\For{$s=1:T$}
	\State $\tilde{D}_N \gets N$ datapoints sampled (with replacement) from $D_N$
    \State $G_s \gets \texttt{DAGLearner}(D_N)$ e.g. \citep{IGSP,IMEC2012} 
    \State $\mathcal{\hat{G}}_T \gets \mathcal{\hat{G}}_T \cup G_s$
\EndFor
\Return $\mathcal{\hat{G}}_T$
\end{algorithmic}
\end{algorithm}

Given $\mathcal{\hat{G}}_T$, which can be constructed from \cref{alg:sample_dags} or sampled from a Markov chain, we next discuss how to compute the expectation over $\theta$. Recall that $U_{\text{M.I.}}^f(\xi; D)$ is given by
\begin{equation} \label{eq:theta_1}
\begin{split}
& \E_{G \mid D} \left[ \E_{y \mid G, \xi} U^f_{\text{M.I.}}(y, \xi; D) \right] \\
	&= \E_{G \mid D} \left[ \E_{\theta \mid G, D} \E_{y \mid G, \theta, \xi} U^f_{\text{M.I.}}(y, \xi; D) \right].
\end{split}
\end{equation}
Instead of carrying out the expensive expectation over $\theta \mid G, D$ in \cref{eq:theta_1}, we use the MLE of $\theta$ for each sampled $G$. This approximation is justified by the \emph{Bernstein-von Mises Theorem}, which implies that  
%
%
\begin{equation} \label{eq:bvm}
\begin{split}
    \pr(\mathbb{\theta} \mid G, D) &\rightarrow N(\mle, \frac{1}{n} I(\theta_G)^{-1}),\\
    \mle &\coloneqq \argmax{\theta} \ \pr(D \mid G, \theta).
\end{split}
\end{equation}
Here, $n$ is the number of datapoints in $D$, and $I(\theta_G)$ is the Fisher information matrix of the parameter $\theta_G$, which is the asymptotic limit of the maximum likelihood estimator $\mle$ \citep[Chapter 10]{van2000asymptotic}. Therefore, the posterior distribution $\theta \mid D, G$ concentrates around $\mle$ at the standard $O(1/\sqrt{n})$ statistical rate. Hence, for moderate $n$ (e.g., when a moderate amount of observational data is provided at the start of the experimental design), \cref{eq:bvm} implies   
\begin{equation} \label{eq:theta}
\begin{split}
& \E_{G \mid D} \E_{\theta \mid G, D} \left[ \E_{y \mid G, \theta, \xi} U^f_{\text{M.I.}}(y, \xi; D) \right] \\
	& \approx \E_{G \mid D} \left[ \E_{y \mid G, \mle, \xi} U^f_{\text{M.I.}}(y, \xi; D) \right].
\end{split}
\end{equation}
In \citet[Section 6.1]{IMEC2015}, the authors provide a closed-form expression for $\mle$ when $y \mid G, \theta$ is multivariate Gaussian. In this case, $\mle$ is a simple function of the sample covariance matrix. 

\subsection{Approximating Mutual Information}
While in the previous subsection we showed how to approximate the expectations in \cref{eq:theta}, computing $U_{\text{M.I.}}^f(y, \xi)$ even for a fixed $y$ is intractable since we must sum over all possible DAGs. Recall from \cref{eq:mutual_info} that the mutual information utility function is    
\begin{equation} \label{eq:repeat_mi}
U_{\text{M.I.}}^f(y, \xi; D) = H(f \mid D) - H(f \mid D, y=y, \xi).
\end{equation}
Note that $H(f \mid D)$ is a constant and does not matter in the optimization over $\xi$. More care is required for  computing the second term in \cref{eq:repeat_mi}, since the posterior of $G$ changes as a result of observing $y$, the realizations of the interventions specified by $\xi$. We therefore cannot immediately use the samples in $\mathcal{\hat{G}}_T$ to approximate this term. To overcome this problem, we propose to use \emph{weighted importance sampling} and approximate $ H(f \mid D, y, \xi)$ by a weighted average of DAGs in $\mathcal{\hat{G}}_T$. We define the importance sample weights for DAG $G_i$, $1 \leq i \leq T$, by
%
%
\begin{equation} \label{eq:is_weights_theta}
w_i \coloneqq \frac{\pr(D, y \mid G_i, \xi)}{\pr(D \mid G_i)}.
\end{equation}
In general, $w_i$ is not equal to $\pr(y \mid G, \xi)$ since $D$ and $y$ are  dependent without conditioning on $\theta$. 
%
%
While $\pr(D, y \mid G, \xi)$ can be computed in closed-form if the prior on $\theta \mid G$ belongs to one of the families described in \citet{bge_orig}, the dependence on previous samples in the importance weights makes greedily building up the intervention set $\xi$ expensive. In particular, since \cref{eq:is_weights_theta} does not factorize, the importance weights must be recomputed with every new additional intervention, which again requires an integration over all parameters. 
Motivated by the approximation in \cref{sec:tract_algo}, where the parameters of each sampled $G \in \mathcal{\hat{G}}_T$ are not marginalized out, we instead propose using the importance sample weights   
\begin{equation} \label{eq:is_modular_weights}
\begin{split}
\hat{w}_i &\coloneqq \frac{\pr(D, y \mid G_i, \hat{\theta}^{G_i}_{\text{MLE}}, \xi)}{\pr(D \mid G_i, \hat{\theta}^{G_i}_{\text{MLE}} )} \\
&= \pr(y \mid G_i, \xi, \hat{\theta}^{G_i}_{\text{MLE}});
\end{split}
\end{equation}
%
$\hat{w}_i$ has the natural interpretation of re-weighting each DAG  by the likelihood of the newly observed data $y$. 

Recall from \cref{eq:mutual_info}, that $U_{\text{M.I.}}^f(y, \xi; D)$ is based on weighting each DAG according to its posterior probability $\pr(G \mid D) \propto \int_{\theta} \pr(D \mid G, \theta) \pr(\theta \mid G) \pr(G)$. Using the importance sample weights $\hat{w}_i$ translates into approximating the mutual information against a different posterior distribution in \cref{eq:mutual_info}, namely 
\begin{equation} \label{eq:importance_post_dist}
  \tilde{\pr}(G \mid D) \propto \pr(D \mid G, \mle) \pr(G), 
\end{equation}
which is a specific instance of an empirical Bayes approximation. In what follows, we denote the mutual information score based on the posterior in \cref{eq:importance_post_dist} by $\tilde{U}_{\text{M.I.}}^f(y, \xi; D)$.

%
%


\subsection{Greedy Optimization}
The cardinality constraint $|\xi| = N_b$ makes our optimization problem a difficult integer program. In the following, we show how to overcome this final computational hurdle using a generalized notion of \emph{submodularity} for multisets \citep{dr_submodular}. In particular, we prove that greedily selecting interventions provides a $(1 - \frac{1}{e})$ guarantee on optimization quality. 
\begin{algorithm} 
\caption{\texttt{GreedyDesign}}\label{alg:greedy_algo}
\hspace*{\algorithmicindent} \textbf{Input:} Utility function $U$, number of samples $N_b$, intervention family $\I^*$  \\
\hspace*{\algorithmicindent} \textbf{Output:} Multiset of interventions $\xi$
\begin{algorithmic}[1]
\State $\xi \gets \emptyset$
\For{$s=1:N_b$}
		\State $I^* \in \argmax{I \in I^*} \ U(\xi \cup I)$
        \State $\xi \gets \xi \cup I^*$
\EndFor
\Return $\xi$
\end{algorithmic}
\end{algorithm}
\begin{nthm}\label{thm:greedy_gaurantee} Suppose $f(G) = G$ i.e. the goal is to recover the full graph as in \citet{koller_active,berger,murphy_active_learning,bayes_sachs}. Then the difference between the global optimum
\begin{equation*}
v^*_b = \max_{\xi \in  \mathbb{Z}^{\I^*} \cap C_{N, b}} \ \E_{G \mid D_{b-1}} \E_{y \mid G, \mle, \xi} \ \tilde{U}_{\text{M.I.}}^f(y, \xi; D) 
\end{equation*}
and $\tilde{v}_b = \texttt{GreedyDesign}(\tilde{U}_{\text{M.I.}}^f, N_b, \I^*)$, the output of \cref{alg:greedy_algo} in batch $b$,
satisfies $\tilde{v}_b \geq (1 - \frac{1}{e})v^*_b$, where $C_{N, b}$ is defined as in \cref{eq:our_constraints}.
\end{nthm}
\begin{rmk}
We conjecture that \cref{thm:greedy_gaurantee} holds for arbitrary functions $f$, but we currently only have a proof (see Appendix) for the case when $f(G) = G$. 
\end{rmk}

We conclude this section by summarizing the developed \emph{Active Budgeted Causal Design Strategy} (\emph{ABCD-Strategy}) in \cref{alg:abc_algo} and then summarizing all the proposed approximations.  
\begin{algorithm}[h]
\label{main_alg}
\caption{\texttt{ABCD}-Strategy}\label{alg:abc_algo}
\hspace*{\algorithmicindent} \textbf{Input:} Target functional $f$, interventional data collected $D_{b-1}$, observational data $D_{obs}$, number of batch samples $N_b$, intervention family $\I^*$, number of DAGs $T$, number of datasets $M$   \\
\hspace*{\algorithmicindent} \textbf{Output:} Multiset of interventions $\xi$
\begin{algorithmic}[1]
\State $\xi \gets \emptyset$
\State $G_T \gets$\texttt{DAGBootSample}($[D_{obs} ,D_{b-1}]$, $T$) 
\State Compute $\hat{U}^f_{\text{M.I.}}$ via \cref{eq:final_approx} \\
\Return \texttt{GreedyDesign}($\hat{U}^f_{\text{M.I.}}$, $N_b$, $\I^*$)
\end{algorithmic}
\end{algorithm}

In terms of approximations, \cref{eq:theta} implies
\begin{equation} \label{eq:final_approx}
\begin{split}
& \E_{G, \theta \mid D} \left[\E_{y \mid G, \theta, \xi}  \tilde{U}^f_{\text{M.I.}}(y, \xi; D) \right] \\
& \approx \E_{G, \mid D} \left[ \E_{y \mid G, \mle, \xi} \ \tilde{U}_{\text{M.I.}}(y, \xi; D) \right] \\
& \approx \sum_{t=1}^T \sum_{m=1}^M \tilde{U}^f_{\text{M.I.}}(y_{tm}, \xi; D), \\
& \qquad \text{s.t. } y_{tm} \overset{\text{i.i.d}}{\sim} y \mid G_t, \mle, \xi \\
& \approx \sum_{t=1}^T \sum_{m=1}^M \hat{U}^f_{\text{M.I.}}(y_{tm}, \xi; D), \ \text{where}
\end{split}
\end{equation}
\begin{equation*}
\begin{split}
& \hat{U}^f_{\text{M.I.}}(y_{tm}, \xi; D) \coloneqq H_{1}(f \mid D) - H_{2}(f \mid D), \\
& \hat{\pr}_{1}(G \mid D) \coloneqq w_{G, D} \ind(G \in \mathcal{\hat{G}}_T), \\
& \hat{\pr}_{2}(G \mid D, y, \xi) \coloneqq \frac{w_{G, D}\pr(y \mid G, \xi, \mle)}{\sum_{t=1}^T w_{G_t, D}\pr(y \mid G_t, \xi, \hat{\theta}^{G_t}_{\text{MLE}})}, 
\end{split}
\end{equation*}
where $M$ is the number of synthetic datasets generated, $H_{1}$ and $H_{2}$ are the entropies induced by $\hat{\pr}_1$ and $\hat{\pr}_2$ respectively, and $w_{G, D}$ is defined in \cref{eq:post_weights}. Note that $\hat{U}^f_{\text{M.I.}}$ is based on the importance sample weights given in \cref{eq:is_modular_weights}.
\begin{nprop} \label{prop:total_runtime}
The total runtime of \cref{alg:greedy_algo} with input utility function $ \hat{U}^f_{\text{M.I.}}$ is $O(pT\kappa^3 + |\I^*| MT^2N_b\kappa p)$, where $\kappa$ is the maximum indegree of a graph in $\mathcal{\hat{G}}_T$. 
\end{nprop}
See \cref{A:runtime_proof} for the proof of \cref{prop:total_runtime}.
\section{Experiments} \label{sec:experiments}

We begin by considering a simple case to demonstrate the behavior of our ABCD-strategy under easily interpretable conditions. Consider the chain graph on $2m -1$ nodes,
\begin{align}
1 \rightarrow 2 \rightarrow \ldots \rightarrow m \rightarrow \ldots \rightarrow p=2m-1.
\end{align}\label{eq:line-graph}
The corresponding essential graph is completely undirected, and the MEC has $2m - 1$ members, one with each node as the source. Assume that sufficient observational data is available to identify the MEC, and we are interested in fully identifying the DAG. Then, our ABCD-strategy selects interventions in order to minimize the expected entropy of the posterior over this MEC. Given a limit of one intervention per batch but infinite samples per batch, \cref{prop:infinite_samps} implies the expected entropy after intervening at node $i$ or $2m - i$, $1 \leq i \leq m$, is 
$$ \frac{1}{2m-1} \Big( \sum_{j < i} \log (i-2) + \sum_{j > i} \log (m-(i+2)) \Big),$$
\begin{figure}[b]
\centering
\includegraphics[width=.4\textwidth]{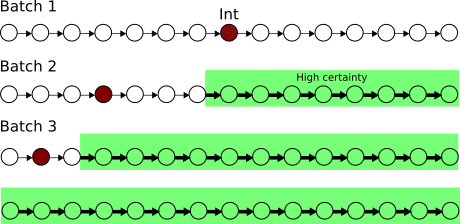}
\caption{Illustration of active learning on a chain graph, beginning with a known MEC on a simulated dataset with $p=15$ nodes. The brown circles indicate the interventions selected in each batch.}
\vspace*{-.1in}
\label{bisection}
\end{figure}%
which is minimized by choosing the midpoint $i = m$. Analogously, we see that the updated $\{\emptyset, \{m\}\}$-MEC is of the same form, so in the second batch, the optimal intervention will be halfway through the remaining nodes. This process of bisection is illustrated in Figure~\ref{bisection} and matches the behavior of our algorithm even in the finite-sample regime as described next.

Figure \ref{boxplot} illustrates the performance of our ABCD-strategy on fifty 11-node chain graphs with random edge weights sampled from $[-1, -.25] \cup [.25, 1]$. For comparison, we consider a random intervention strategy that uniformly distributes the samples in each batch to $k$ interventions picked uniformly at random, where $k$ is the maximum number of unique interventions allowed per batch. Whereas the median-performing random strategy barely reduces the entropy, the ABCD-strategy reduces the entropy significantly in all runs. When all $k$ interventions are picked for the same batch, so that ABCD receives no feedback, the median-performing run of active learning still reduces the entropy as much as the best-performing runs of the random strategy. 

\begin{figure}[!t]
\centering
\includegraphics[width=.35\textwidth]{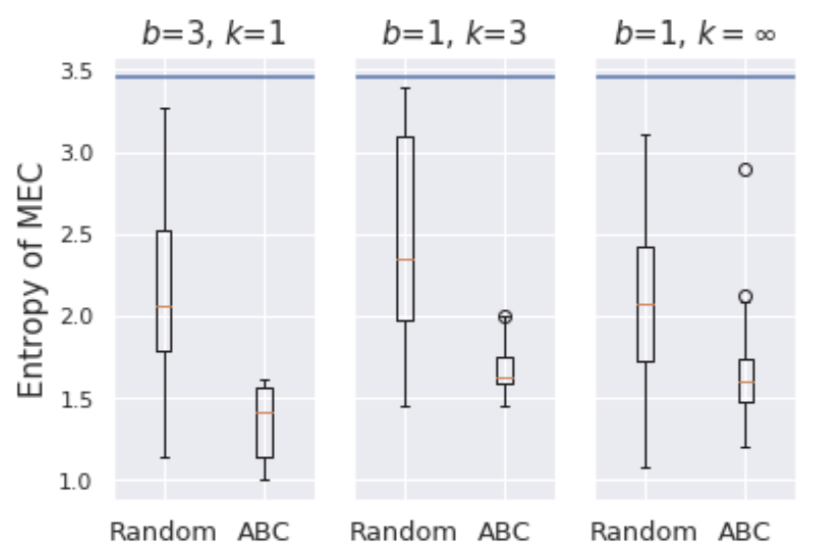}
\caption{Box plots for 50 runs of the random strategy versus our ABCD-strategy on the graph in Figure~\ref{eq:line-graph} with $p = 11$ and $n = 30$ samples. The horizontal line indicates the entropy of the prior distribution, i.e.~uniform over the MEC. Note that $k=\infty$ corresponds to the case with  no constraints on the number of unique interventions.}
\label{boxplot}
\end{figure}

Having demonstrated the behavior of ABCD for a simple case, we now analyze the performance of our method on more general DAGs. The skeleton of each graph is sampled from an Erd\"{o}s-R\'{e}nyi model with density $\rho=0.25$. The edges of these graphs are directed by sampling a permutation of the nodes uniformly at random and orienting the edges accordingly. To avoid long runtimes when enumerating the MEC, we disposed of graphs with more than 100 members in their MEC.\footnote{From a sample of 10,000 graphs, only 54 had MEC size greater than 100. Based on the results by \citet{mec_size}, we expect the MECs to be typically small.} When the MEC is known, we may define a variant of the random strategy, \emph{Chordal-Random}, which only intervenes on nodes that are in chordal components of the essential graph, i.e., nodes adjacent to at least one undirected edge. Since the Meek rules can only propagate by intervening within chordal components, Chordal-Random is a more fair baseline strategy for comparison than simple random sampling.   
\begin{figure}
\centering
\begin{subfigure}{.35\textwidth}
\centering
\includegraphics[width=\textwidth]{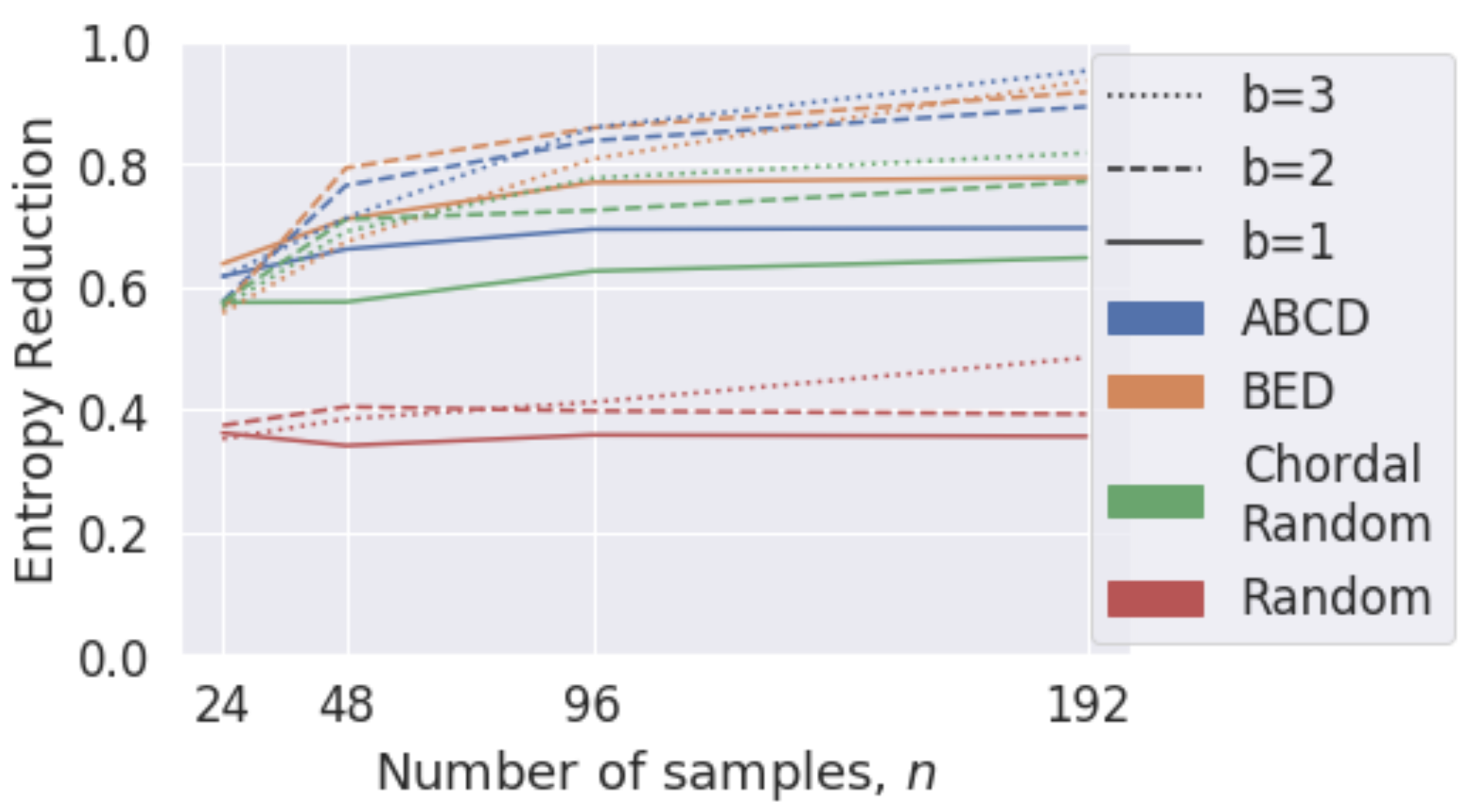}
\caption{Single MEC}\label{fig:entropy_reduction}
\end{subfigure}

\begin{subfigure}{.38\textwidth}
\centering
\includegraphics[width=\textwidth]{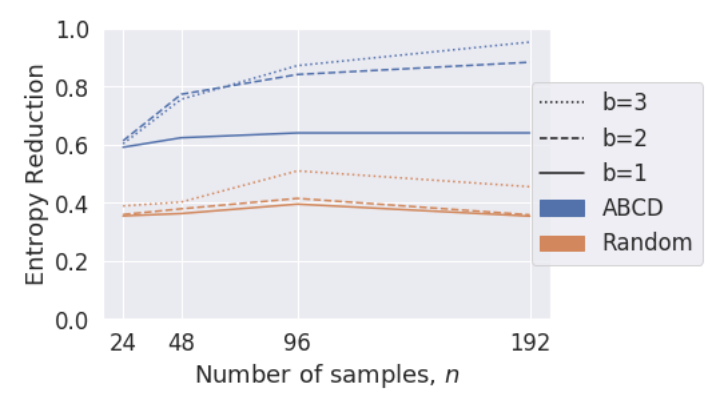}
\caption{Multiple MECs.}\label{fig:entropy_reduction_multiple}
\end{subfigure}

\caption{Performance of intervention strategies for batch sizes $b$ as a function of the total number of samples, computed from 50 Erd\"{o}s-R\'{e}nyi DAGs with density $\rho = 0.25$.}
\vspace*{-.2cm}
\end{figure}

Figure \ref{fig:entropy_reduction} demonstrates the improvement in selecting interventions using the ABCD-strategy as compared to Chordal-Random when the number of unique interventions per batch is bounded by one. The entropy reduction for an interventional data set $D_{\xi}$ is defined as $\frac{H(G) - H(G | D_{\xi})}{H(G)}$, and it is used as a metric so that MECs of different sizes are comparable. Since the number of total possible unique interventions is $kB$, an increase in the number of batches also increases the variability of the interventions, reflected in the increase of entropy reduction with batch size. Already with only 192 samples and 3 total batches, our ABCD-strategy is able to learn most graphs with complete certainty. The comparable performance of the Budgeted Experiment Design (BED) strategy \citep{bed_MEC} suggests that for the given experimental setup, the interventions that orient the most edges correspond well to those that most reduce entropy as we discussed in \cref{prop:infinite_samps}. Figure \ref{fig:entropy_reduction_multiple} shows that the performance of the ABCD-strategy remains strong even when the MEC of the graph is not known. Specifically, up to 3 additional MECs were generated by randomly flipping non-covered edges that did not create cycles, and again only graphs for which the union of these MECs had cardinality less than 100 were kept. Note that we are not able to compare with BED since BED requires that the MEC is known.
\begin{figure}[!t]
\centering
\includegraphics[width=.2\textwidth]{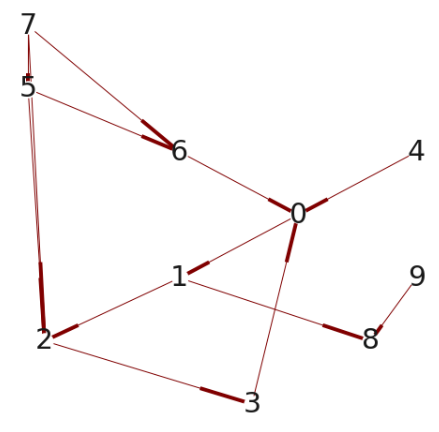}
\includegraphics[width=.3\textwidth]{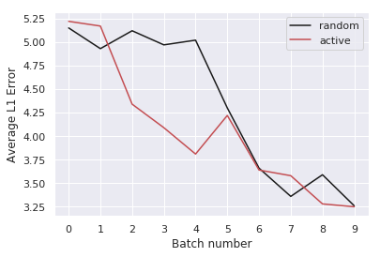}
\caption{Top: DREAM4 ground truth 10-node network. Bottom: Performance of intervention strategies on predicting the descendants of gene $0$.}
\label{fig:dream4}
\vspace*{-.2cm}
\end{figure}

{\bf DREAM4 Synthetic Dataset.} Finally, we applied our experimental design strategy to gene expression data from the DREAM4 10-node in-silico network reconstruction challenge \citep{dream4_data}. These data are generated from stochastic differential equations and simulate microarray data of gene regulatory networks. We constructed an observational dataset from the wild-type, multifactorial perturbation, and time $0$ time-series samples (16 samples in total), and similarly, interventional datasets from the knockdown and knockout samples (2 samples each).

Previous work on experimental design 
applied to biological datasets \citep{berger} has focused on learning the entire network. In practice, practitioners may be specifically interested in performing experiments to elucidate a functional of the network, such as the pathway or local network surrounding a gene of interest. To emulate this setting, we applied our ABCD-strategy towards learning the \emph{downstream genes} of select genes from the true network (Figure \ref{fig:dream4}, top). Despite high variations in learning due to the small size of the dataset, we observed an improvement over the random strategy for several central genes (\cref{fig:dream4}, bottom; \cref{fig:dream4-supp}). These results illustrate the promise of applying targeted experimental design for applications to genomics.

\section{Concluding Remarks}
We proposed \emph{Active Budgeted Causal Design Strategy} (\emph{ABCD-Strategy}), an experimental method based on optimal Bayesian experimental design with provable guarantees on approximation quality. Empirically, we demonstrated that ABCD yields considerable boosts over random sampling for both targeted and full causal structure discovery. 
Such experimental design strategies are particularly relevant for applications to genomics, where the number of possible experiments is huge due to the possibility of intervening on combinations of genes.

\section*{Acknowledgements}

R.~Agrawal was partially supported by IBM. K.D. Yang was supported by an NSF graduate fellowship and ONR (N00014-18-1-2765).  C.~Uhler was partially supported by NSF (DMS-1651995), ONR (N00014-17-1-2147 and N00014-18-1-2765), IBM, and a Sloan Fellowship. 

\bibliographystyle{abbrvnat} 
\bibliography{references}

\begin{thebibliography}{36}
\providecommand{\natexlab}[1]{#1}
\providecommand{\url}[1]{\texttt{#1}}
\expandafter\ifx\csname urlstyle\endcsname\relax
  \providecommand{\doi}[1]{doi: #1}\else
  \providecommand{\doi}{doi: \begingroup \urlstyle{rm}\Url}\fi

\bibitem[Agrawal et~al.(2018)Agrawal, Broderick, and Uhler]{minimap_paper}
R.~Agrawal, T.~Broderick, and C.~Uhler.
\newblock Minimal {I-MAP MCMC} for scalable structure discovery in causal {DAG}
  models.
\newblock In \emph{International Conference on Machine Learning}, 2018.

\bibitem[Andersson et~al.(1997)Andersson, Madigan, and Perlman]{andersson97}
S.~A. Andersson, D.~Madigan, and M.~D. Perlman.
\newblock A characterization of {Markov} equivalence classes for acyclic
  digraphs.
\newblock \emph{Annals of Statistics}, 25\penalty0 (2):\penalty0 505--541,
  1997.

\bibitem[Chaloner and Verdinelli(1995)]{bayes_exp_design}
K.~Chaloner and I.~Verdinelli.
\newblock {Bayesian} experimental design: A review.
\newblock \emph{Statistical Science}, 10:\penalty0 273--304, 1995.

\bibitem[Cho et~al.(2016)Cho, Berger, and Peng]{berger}
H.~Cho, B.~Berger, and J.~Peng.
\newblock Reconstructing causal biological networks through active learning.
\newblock \emph{PLoS ONE}, 2016.

\bibitem[Dixit et~al.(2016)Dixit, Parnas, Li, Chen, Fulco, Jerby-Arnon,
  Marjanovic, Dionne, Burks, Raychowdhury, Adamson, Norman, Lander, Weissman,
  Friedman, and Regev]{perturb_seq}
A.~Dixit, O.~Parnas, B.~Li, J.~Chen, C.~Fulco, L.~Jerby-Arnon, N.~Marjanovic,
  D.~Dionne, T.~Burks, R.~Raychowdhury, B.~Adamson, T.~Norman, E.~Lander,
  J.~Weissman, N.~Friedman, and A.~Regev.
\newblock Perturb-seq: dissecting molecular circuits with scalable single-cell
  {RNA} profiling of pooled genetic screens.
\newblock \emph{Cell}, pages 1853--1866, 2016.

\bibitem[Ellis and Wong(2008)]{ellis_wong}
B.~Ellis and W.~H. Wong.
\newblock Learning causal {Bayesian} network structures from experimental data.
\newblock \emph{Journal of the American Statistical Association}, 103:\penalty0
  778--789, 2008.

\bibitem[Friedman and Koller(2003)]{Friedman2003}
N.~Friedman and D.~Koller.
\newblock Being {Bayesian} about network structure. {A} {Bayesian} approach to
  structure discovery in {Bayesian} networks.
\newblock \emph{Machine Learning}, 50:\penalty0 95--125, 2003.

\bibitem[Friedman et~al.(1999)Friedman, Goldszmidt, and Wyner]{bootstrap_dag}
N.~Friedman, M.~Goldszmidt, and A.~J. Wyner.
\newblock Data analysis with {Bayesian} networks: A bootstrap approach.
\newblock In \emph{Proceedings of the Fifteenth Conference on Uncertainty in
  Artificial Intelligence}, 1999.

\bibitem[Friedman et~al.(2000)Friedman, Linial, Nachman, and Pe'er]{friedman00}
N.~Friedman, M.~Linial, I.~Nachman, and D.~Pe'er.
\newblock Using {Bayesian} networks to analyze expression data.
\newblock \emph{Journal of Computational Biology}, 7\penalty0 (3-4):\penalty0
  601--620, 2000.

\bibitem[Geiger and Heckerman(1999)]{bge_orig}
D.~Geiger and D.~Heckerman.
\newblock Parameter priors for directed acyclic graphical models and the
  characterization of several probability distributions.
\newblock In \emph{Proceedings of the Fifteenth Conference on Uncertainty in
  Artificial Intelligence}, 1999.

\bibitem[Ghassami et~al.(2018)Ghassami, Salehkaleybar, Kiyavash, and
  Bareinboim]{bed_MEC}
A.~Ghassami, S.~Salehkaleybar, N.~Kiyavash, and E.~Bareinboim.
\newblock Budgeted experiment design for causal structure learning.
\newblock In \emph{International Conference on Machine Learning}, 2018.

\bibitem[Gillispie and Perlman(2001)]{mec_size}
S.~B. Gillispie and M.~D. Perlman.
\newblock Enumerating {Markov} equivalence classes of acyclic digraph models.
\newblock In \emph{Proceedings of the 17th Conference in Uncertainty in
  Artificial Intelligence}, 2001.

\bibitem[Grzegorczyk and Husmeier(2008)]{edge_reversal}
M.~Grzegorczyk and D.~Husmeier.
\newblock Improving the structure {MCMC} sampler for {Bayesian} networks by
  introducing a new edge reversal move.
\newblock \emph{Machine Learning}, 71:\penalty0 265--305, 2008.

\bibitem[Hauser and B\"{u}hlmann(2012)]{IMEC2012}
A.~Hauser and P.~B\"{u}hlmann.
\newblock Characterization and greedy learning of interventional {Markov}
  equivalence classes of directed acyclic graphs.
\newblock \emph{Journal of Machine Learning Research}, 13\penalty0
  (1):\penalty0 2409--2464, 2012.

\bibitem[Hauser and B{\"{u}}hlmann(2014)]{hauser_MEC}
A.~Hauser and P.~B{\"{u}}hlmann.
\newblock Two optimal strategies for active learning of causal models from
  interventional data.
\newblock \emph{International Journal of Approximate Reasoning}, 55:\penalty0
  926--939, 2014.

\bibitem[Hauser and B\"{u}hlmann(2015)]{IMEC2015}
A.~Hauser and P.~B\"{u}hlmann.
\newblock Jointly interventional and observational data: estimation of
  interventional {Markov} equivalence classes of directed acyclic graphs.
\newblock \emph{Journal of the Royal Statistical Society Series B}, 77\penalty0
  (1):\penalty0 291--318, 2015.

\bibitem[Heckerman et~al.(1997)Heckerman, Meek, and Cooper]{high_prob_dags}
D.~Heckerman, C.~Meek, and G.~Cooper.
\newblock A {Bayesian} approach to causal discovery.
\newblock Technical report, Microsoft Research, 1997.

\bibitem[Kuipers and Moffa(2017)]{partition_mcmc}
J.~Kuipers and G.~Moffa.
\newblock Partition {MCMC} for inference on acyclic digraphs.
\newblock \emph{Journal of the American Statistical Association}, 112:\penalty0
  282--299, 2017.

\bibitem[Kuipers et~al.(2014)Kuipers, Moffa, and Heckerman]{bge_score}
J.~Kuipers, G.~Moffa, and D.~Heckerman.
\newblock Addendum on the scoring of {Gaussian} directed acyclic graphical
  models.
\newblock \emph{The Annals of Statistics}, 42:\penalty0 1689--1691, 2014.

\bibitem[Lauritzen(1996)]{lauritzen_book}
S.~Lauritzen.
\newblock \emph{Graphical Models}.
\newblock Oxford University Press, 1996.

\bibitem[Madigan and York(1995)]{struct_mcmc}
D.~Madigan and J.~York.
\newblock {Bayesian} graphical models for discrete data.
\newblock \emph{International Statistical Review}, 63:\penalty0 215--232, 1995.

\bibitem[Murphy(2001)]{murphy_active_learning}
K.~Murphy.
\newblock Active learning of causal {Bayes} net structure.
\newblock Technical report, 2001.

\bibitem[Ness et~al.(2018)Ness, Sachs, Mallick, and Vitek]{bayes_sachs}
R.~O. Ness, K.~Sachs, P.~Mallick, and O.~Vitek.
\newblock A {Bayesian} active learning experimental design for inferring
  signaling networks.
\newblock \emph{Journal of Computational Biology}, 25\penalty0 (7):\penalty0
  709--725, 2018.

\bibitem[Niinimaki et~al.(2016)Niinimaki, Parviainen, and
  Koivisto]{partial_order}
T.~Niinimaki, P.~Parviainen, and M.~Koivisto.
\newblock Structure discovery in {Bayesian} networks by sampling partial
  orders.
\newblock \emph{Journal of Machine Learning Research}, 17:\penalty0 2002--2048,
  2016.

\bibitem[Pearl(2003)]{pearl03}
J.~Pearl.
\newblock Causality: Models, reasoning, and inference.
\newblock \emph{Econometric Theory}, 19\penalty0 (675-685):\penalty0 46, 2003.

\bibitem[Robins et~al.(2000)Robins, Hernan, and Brumback]{robins00}
J.~M. Robins, M.~A. Hernan, and B.~Brumback.
\newblock Marginal structural models and causal inference in epidemiology,
  2000.

\bibitem[Schaffter et~al.(2011)Schaffter, Marbach, and Floreano]{dream4_data}
T.~Schaffter, D.~Marbach, and D.~Floreano.
\newblock {GeneNetWeaver}: in silico benchmark generation and performance
  profiling of network inference methods.
\newblock \emph{Bioinformatics}, 27:\penalty0 2263--2270, 2011.

\bibitem[Soma and Yoshida(2016)]{dr_submodular}
T.~Soma and Y.~Yoshida.
\newblock Maximizing monotone submodular functions over the integer lattice.
\newblock In \emph{International Conference on Integer Programming and
  Combinatorial Optimization}, pages 325--336. Springer, 2016.

\bibitem[Soma et~al.(2014)Soma, Kakimura, Inaba, and
  Kawarabayashi]{submod_budget_alloc}
T.~Soma, N.~Kakimura, K.~Inaba, and K.~Kawarabayashi.
\newblock Optimal budget allocation: Theoretical guarantee and efficient
  algorithm.
\newblock In \emph{International Conference on International Conference on
  Machine Learning}, 2014.

\bibitem[Spirtes et~al.(2000)Spirtes, Glymour, and Scheines]{causality_book}
P.~Spirtes, C.~Glymour, and R.~Scheines.
\newblock \emph{Causation, Prediction, and Search}.
\newblock MIT press, 2nd edition, 2000.

\bibitem[Tong and Koller(2001)]{koller_active}
S.~Tong and D.~Koller.
\newblock Active learning for structure in {Bayesian} networks.
\newblock In \emph{International Joint Conference on Artificial Intelligence},
  2001.

\bibitem[Van~der Vaart(2000)]{van2000asymptotic}
A.~W. Van~der Vaart.
\newblock \emph{Asymptotic statistics}, volume~3.
\newblock Cambridge university press, 2000.

\bibitem[Verma and Pearl(1991)]{verma90}
T.~S. Verma and J.~Pearl.
\newblock Equivalence and synthesis of causal models.
\newblock In \emph{Uncertainty in Artificial Intelligence}, volume~6, page 255,
  1991.

\bibitem[Verma and Pearl(1992)]{meek_rules}
T.~S. Verma and J.~Pearl.
\newblock An algorithm for deciding if a set of observed independencies has a
  causal explanation.
\newblock In \emph{Uncertainty in Artificial Intelligence}, 1992.

\bibitem[Wang et~al.(2017)Wang, Solus, Yang, and Uhler]{IGSP}
Y.~Wang, L.~Solus, K.~Yang, and C.~Uhler.
\newblock Permutation-based causal inference algorithms with interventions.
\newblock In \emph{Advances in Neural Information Processing Systems}, pages
  5824--5833, 2017.

\bibitem[Yang et~al.(2018)Yang, Katcoff, and Uhler]{yang2018characterizing}
K.~D. Yang, A.~Katcoff, and C.~Uhler.
\newblock Characterizing and learning equivalence classes of causal {DAGs}
  under interventions.
\newblock In \emph{International Conference on Machine Learning}, 2018.

\end{thebibliography}

\clearpage

\appendix

\section{Proofs}

\subsection{Proof of \cref{prop:infinite_samps}}
Given infinite samples per intervention $I \in \I$, $G^*$ is recovered up to its $\I$-Markov equivalence class. Hence, the resulting entropy after placing an infinite number of samples at each intervention is equal to $\log_2 |\text{Ess}^{\I}(G)|$ when the true DAG is $G$. Since the true DAG is unknown, this entropy must be averaged over our prior distribution on $\mathcal{G}$, which is uniform. Hence, the entropy after observing an infinite number of samples per intervention in $\I$ equals $\frac{1}{|\mathcal{G}|} \sum_{G \in \mathcal{G}}  \log_2 |\text{Ess}^{\I}(G)|$. Minimizing this entropy over all possible interventions sets of size at most $K$ completes the proof. 
\subsection{Proof of \cref{thm:mutal_consistent}}

Let
\begin{equation*}
\I^{\infty} \coloneqq \{I \in \I^*: \sum_{b=1}^{\infty} |\tilde{I} \in \xi_{b}: \tilde{I} = I | = \infty \ \mu^* a.s. \},
\end{equation*}
where $\xi_{b}$ denotes the interventions selected at batch $b$ by $U_{\text{M.I}}^f$. Since $\I^*$ is finite, $\I^{\infty}$ is non-empty. When $|\I^{\infty}| > 1$, $\I^{\infty}$ is a conservative family of targets since $\I^*$ is a family of single-node interventions. Hence, we identify the $\I^{\infty}$-MEC of $G^*$ in the limit of an infinite number of batches and samples \citep{IMEC2012}. Assume $|\I^{\infty}| > 1$. If $f(G)$ is identifiable in $\text{Ess}^{\I^{\infty}}(G^*)$, then 
\begin{equation*}
\pr \left( f(G) \mid D_{B} \right) \xrightarrow{\text{$\mu^*$ a.s.}} \ind(f(G) = f(G^*)). 
\end{equation*}
Hence, it suffices to show that the interventions $U_{\text{M.I}}^f$ selects infinitely often identifies $f(G)$ in the limiting interventional essential graph $\text{Ess}^{\I^{\infty}}(G^*)$. Suppose towards a contradiction that $f(G)$ were not fully identifiable in $\text{Ess}^{\I^{\infty}}(G^*)$. 
By definition of almost sure convergence, there exists some $b^* < \infty$ such that any $\tilde{I} \in \I^* \setminus \I^{\infty}$ is never selected again after batch $b^*$ with probability one since $\I^*$ is finite. Maximizing $U^f_{\text{M.I.}}$ is equivalent to minimizing the conditional entropy,
\begin{equation}
H_{\xi}^b(f \mid Y_{\xi}) \coloneqq \E_{y \sim \pr(y \mid D_{b}, \xi)} \ H(f \mid D_{b}, Y=y).  
\end{equation}
If $b > b^*$, then  
\begin{equation} \label{eq:batch_opt_contra}
\argmin{\xi \in  \mathbb{Z}^{\I^*} \cap  C_b} \ H_{\xi}^b(f \mid Y_{\xi}) = \argmin{\xi \in  \mathbb{Z}^{\I^{\infty}} \cap  C_b} \ H_{\xi}^b(f \mid Y_{\xi})
\end{equation}
since any batch $b$ after $b^*$ never selects an intervention in $\tilde{I} \in \I^* \setminus \I^{\infty}$. Since $f$ is not identifiable in $\text{Ess}^{\I^{\infty}}(G^*)$, that implies
\begin{equation*}
\lim_{b \rightarrow \infty} H_{\xi_{\infty}}^b(f \mid Y_{\xi_{\infty}}) \rightarrow L > 0. 
\end{equation*}
Since $\I^*$ consists of all single-node interventions, $\I^*$ can identify $f(G)$ \citep{IMEC2012}. Hence, there must be some $\tilde{I} \in \I^* \setminus \I^{\infty}$ and $\epsilon > 0$ such that 
\begin{equation} \label{eq:limit_entropy}
\lim_{b \rightarrow \infty} H_{\xi_{\infty} \cup \tilde{I}_{\infty}}^b(f) < L - \epsilon, 
\end{equation}
where $\tilde{I}_{\infty}$ denotes selecting $\tilde{I}$ infinitely many times. But \cref{eq:limit_entropy} implies that there must exist some batch $b > b^*$ such that the conditional entropy of the design $\tilde{\xi} = \{ \tilde{I} \}$ is uniformly smaller than the conditional entropy of any $\xi \in \mathbb{Z}^{\I^{\infty}}$. But this is a contradiction because then $\tilde{I}$ would be selected again after some batch $b > b^*$ and \cref{eq:batch_opt_contra} would no longer hold. 
%

%

For $|\I^{\infty}| = 1$, we no longer have a conservative family of targets. However, a nearly identical argument works by noting that, in the limit, we learn the observational equivalence class of the $\I^{\infty}$ mutilated graph of $G^*$.

\subsection{Consistency Counterexample} \label{A:consis_counter}
\begin{figure}[t]
\centering
 \includegraphics[width=6cm]{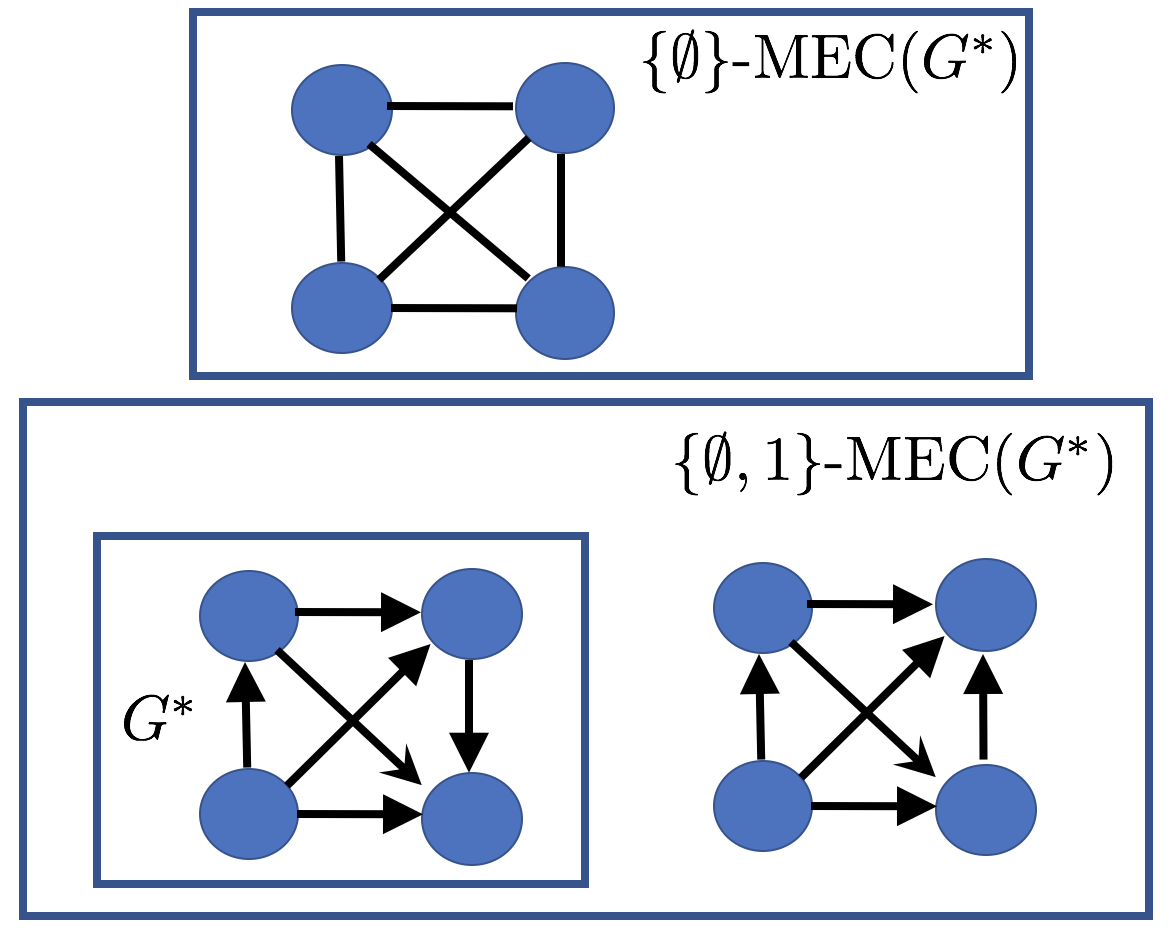}
 \caption{Each box represents the members of the interventional Markov equivalence classes. For $G^*$ given in the bottom left box, the observational Markov equivalence class has no edges oriented. The top box represents the essential graph of the observational Markov equivalence class. The interventional Markov equivalence class for an intervention at node one consists of two DAGs given in the bottom box.}
 \label{fig:consis_counter}
\end{figure}
Suppose we know the Markov equivalence class of $G^*$ and the goal is to fully recover $G^*$. Suppose $C_{b} = \{ \xi: \| \xi\|_0 = K \}$, where $\| \cdot \|_0$ counts the number of unique interventions in $\xi$. Since there is no constraint on the number of samples, only on the number of unique interventions, we may allocate an infinite number of samples per intervention within each batch. This constraint is equivalent to the one examined in \citet{bed_MEC}. The scores in both \citet{bayes_sachs} and \citet{bed_MEC} select interventions by maximizing the expected number of oriented edges in the interventional Markov equivalence classes. In particular, the utility function in \citet{bayes_sachs} is equivalent to maximizing,
\begin{equation} \label{ref:meek_score}
U(\I; D) = \sum_{G \in \mathcal{G}} A(\text{Ess}^{\I}(G)) \pr(G),
\end{equation}
%
where $A(\text{Ess}^{\I}(G))$ equals the additional number of edges oriented relative to the observational Markov equivalence class. Suppose $G^*$ equals  the graph in \cref{fig:consis_counter} and that $K=1$ unique interventions are allowed within each batch. Assume that $\I^* = \{ \{1\}, \cdots, \{4\} \}$ and that we start with a uniform prior over $\mathcal{G}$. Then, since all arrows are undirected in the observational Markov equivalence class, symmetry implies $U(\{j\}; \emptyset) = U(\{j\}; \emptyset)$ for all $i, j \in 1, \cdots, 4$. Without any loss of generality suppose intervention one is selected in batch one. We show that every subsequent batch will select intervention $\{1\}$. If only $\{1\}$ were selected, $U(\I; D)$ would not be consistent since the $\{\emptyset, \{1 \}\}$-MEC($G^*$) contains two graphs, as shown at the bottom of \cref{fig:consis_counter}. After batch one, the posterior is supported on these two graphs since an infinite number of samples are allocated to the intervention at node one. 

The utility function in \cref{ref:meek_score} scores interventions relative to the observational equivalence class, which causes the consistency issue. In particular, the posterior in batch two is only supported on the two DAGs given in the bottom box of \cref{fig:consis_counter}.  The score of $\{1\}$ equals $5$ in batch two while the scores of interventions $\{2\}, \{3\}, \{4\}$ equal $4, 3, 4$, respectively. Hence, in batch two, intervention $\{1\}$ will be selected again, but the posterior will remain the same since the $\{\emptyset, \{1 \}\}$ interventional Markov equivalence class of $G^*$ is already known.

An easy way to fix \cref{ref:meek_score} (for this given counterexample) would be to only select interventions not selected in previous batches. This modification would fix the issue with the counterexample, namely prevent intervention one from being selecting infinitely often. However, when one can only allocate a finite number of samples per batch, this modification would not lead to a consistent estimator. In particular, if a certain intervention is done in some batch, and that intervention must be conducted in order to identify $f$, then only placing finitely many samples to that intervention in that batch and never placing any more samples in subsequent batches will not lead to a consistent method.

\subsection{Proof of \cref{thm:greedy_gaurantee}}
%
%
\begin{ndefn} \citep{dr_submodular}
Let $E$ be a finite set. A function $f: \mathbb{Z}^E \rightarrow \R$ is \emph{diminishing returns submodular} (DR-submodular) if for $x \leq y$ 
\begin{equation}
f(x + \chi_e) - f(x) \geq f(y + \chi_e) - f(y), \ x, y \in \mathbb{Z}^E
\end{equation}
where $e \in E$ and $\chi_e$ is the ith unit vector.  
\end{ndefn}

\begin{nlem} \label{nlem:mut_submod}
$\tilde{U}_{\text{M.I.}}^f(\xi; D)$ is DR-submodular.
\end{nlem}
\begin{proof}
$f(G) = G$ so we omit $f$ in $\tilde{U}_{\text{M.I.}}^f$ to simplify notation. Since the sum of submodular functions is submodular, it suffices to show 
\begin{equation}
\begin{split}
\E_{y \mid G, \mle, \xi} \ \tilde{U}_{\text{M.I.}}(y, \xi; D) & = H(G) - H(G \mid Y_{\xi}) \\
			  & = I((G, \mle), Y_{\xi})
\end{split}
\end{equation}
is DR-submodular, where $I$ is the mutual information. Consider an $A \subseteq B \in \mathbb{Z}^{\I^*}$. Take any $C \in \I^*$. Since entropy decreases with more conditioning,  
\begin{equation} \label{eq:entropy_cond}
    \begin{split}
        H(Y_C \mid Y_A) - H(Y_C \mid (G, \mle)) \geq \\ H(Y_C \mid Y_B) - H(Y_C \mid (G, \mle)).
    \end{split}
\end{equation}
By conditional independence, 
\begin{equation}
    \begin{split}
        H(Y_C \mid (G, \mle)) & = H(Y_C \mid (G, \mle), Y_A) \\
                              &=  H(Y_C \mid (G, \mle), Y_B).
    \end{split}
\end{equation}
Hence, \cref{eq:entropy_cond} may be rewritten as,
\begin{equation} \label{eq:entropy_cond2}
    \begin{split}
        I((G, \mle), Y_C \mid Y_A) = \\ H(Y_C \mid Y_A) - H(Y_C \mid (G, \mle), Y_A) \geq \\ H(Y_C \mid Y_B) - H(Y_C \mid (G, \mle), Y_B) = \\
        I((G, \mle), Y_C \mid Y_B).
    \end{split}
\end{equation}
\cref{eq:entropy_cond2} implies 
\begin{equation}
\begin{split}
& I((G, \mle), Y_{A} + Y_C) - I((G, \mle), Y_{A}) \\
	& \geq I((G, \mle), Y_{B} + Y_C) - I((G, \mle), Y_{B})
\end{split}
\end{equation}
as desired.
\end{proof}

The proof of \cref{thm:greedy_gaurantee} then follows directly from \cref{nlem:mut_submod} and \citet[Theorem 2.4]{submod_budget_alloc}.

\subsection{Proof of \cref{prop:total_runtime}} \label{A:runtime_proof}
For each graph $G \in \mathcal{G}_T$, compute the associated edge weights $\mle$. Computing each $\mle$ takes $O(p\kappa^3)$ time using the formula given in \citet[pg. 17]{IMEC2012}. Since there are $T$ DAGs, the total time to compute the MLE estimates of the edge weights of each DAG is $O(Tp\kappa^3)$. 
Sampling from a multivariate Gaussian with bounded indegree with known adjacency matrix takes $O(p\kappa)$ time. $\hat{U}^f_{\text{M.I.}}$ requires a total of $|\I^*|MN_bT^2$ samples. Hence, the total computation time of sampling all the $y_{mt}$ in \cref{eq:final_approx} is $O(|\I^*|MN_b \kappa pT^2)$. Evaluating $\hat{U}^f_{\text{M.I.}}$ takes $O(M T^2)$ time using these samples, which is of lower computational complexity than computing $\hat{U}^f_{\text{M.I.}}$. Hence, the total runtime is $O(p\kappa^3 + |\I^*|MN_b \kappa pT^2)$. 

\subsection{Constraint on the Number of Unique Interventions} \label{A:unique_int_constraint}

If we are only allowed to allocate at most $K$ unique interventions per batch, we modify \cref{alg:abc_algo} by allocating $\frac{N_b}{K}$ samples per intervention in \cref{alg:greedy_algo}. 
Once an intervention is selected, that intervention is removed from $I^*$ and another one is greedily selected from the remaining set. With this strategy, \cref{alg:greedy_algo} will terminate after $K$ iterations. Hence, there will be at most $K$ unique interventions as desired. 

\subsection{DREAM4 Supplementary Figures}
We applied our targeted experimental design strategy towards learning the \emph{downstream pathways} of select genes from a 10-node network from the DREAM4 challenge. We observed a modest improvement over the random strategy for some central genes in the network (\cref{fig:dream4-supp}, top). However, the results are subject to high variations (\cref{fig:dream4-supp}, bottom), which we surmise to be due to the small size of the observational dataset. Nevertheless, these preliminary results illustrate the promise of applying targeted experimental design to real, large-scale biological datasets.
\begin{figure*}
\centering
\includegraphics[width=.7\textwidth]{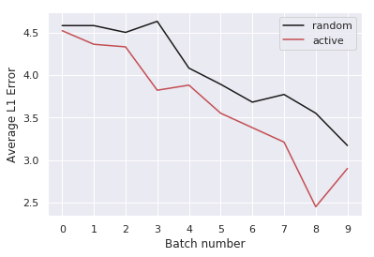}
\includegraphics[width=.7\textwidth]{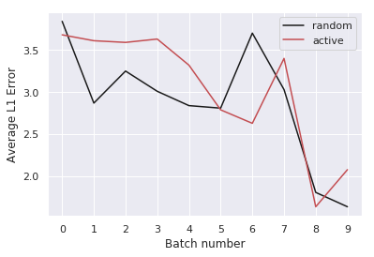}
\caption{Performance of intervention strategies on predicting the descendants of genes $6$ (top) and $8$ (bottom).}
\label{fig:dream4-supp}
\end{figure*}
\end{document}